\documentclass[conference,10pt]{IEEEtran}
\usepackage[utf8]{inputenc}
\usepackage{amsmath}
\usepackage{amsthm}
\usepackage{amsfonts}
\usepackage{mathtools}
\usepackage{tikz}
\usetikzlibrary{patterns}
\usepackage{pgfplots}
\pgfplotsset{compat=1.13}
\usepackage[figure]{algorithm2e}
\usepackage{float}
\usepackage[caption=false]{subfig}
\usepackage{multiaudience}
  \SetNewAudience{conference}
  \SetNewAudience{arxiv}

\newcommand\blankfootnote[1]{%
  \let\svthefootnote\thefootnote%
  \let\thefootnote\relax\footnotetext{#1}%
  \let\thefootnote\svthefootnote%
}

\pgfplotsset{
    discard if/.style 2 args={
        x filter/.code={
            \edef\tempa{\thisrow{#1}}
            \edef\tempb{#2}
            \ifx\tempa\tempb
                
            \fi
        }
    },
    discard if not/.style 2 args={
        x filter/.code={
            \edef\tempa{\thisrow{#1}}
            \edef\tempb{#2}
            \ifx\tempa\tempb
            \else
                
            \fi
        }
    }
}

\definecolor{plotcol1}{RGB}{161,218,180}
\definecolor{plotcol2}{RGB}{65,182,196}
\definecolor{plotcol3}{RGB}{44,127,184}
\definecolor{plotcol4}{RGB}{37,52,148}

\newcommand{\coordinator}{{C}}
\newcommand{\agent}[1]{{A_{#1}}}
\newcommand{\agentSet}{\ensuremath{\mathcal{A}}}
\newcommand{\edgeSet}{\ensuremath{\mathcal{E}}}
\newcommand{\numAgents}{{n}}
\newcommand{\agentIndex}{{k}}
\newcommand{\noise}{{N}}
\newcommand{\agentNoise}[1]{\ensuremath{N_{#1}}}
\newcommand{\channelFading}[1]{\ensuremath{h_{#1}}}
\newcommand{\agentTx}[1]{{\alpha_{#1}}}
\newcommand{\agentRx}[1]{\ensuremath{\Gamma_{#1}}}
\newcommand{\coordinatorRx}{{\gamma}}
\newcommand{\coordinatorTx}{\ensuremath{\beta}}

\newcommand{\lexGreater}{{>}}

\newcommand{\lexGreaterCompatible}{{\geq}}
\newcommand{\agentInputSequence}[1]{{S_{#1}}}
\newcommand{\coordinatorOutputEstimate}{{S}}
\newcommand{\naturals}{{\mathbb{N}}}

\newcommand{\complex}{\ensuremath{\mathbb{C}}}
\newcommand{\infiniteBinarySequences}{{\{0,1\}^\infty}}
\newcommand{\finiteBinarySequences}{{\{0,1\}^{<\infty}}}

\newcommand{\generalBinarySequence}{{S}}
\newcommand{\outputCondition}{{\varphi}}
\newcommand{\outputConditionFreeVariable}{{x}}
\newcommand{\maximumRemainingAgents}{{m}}
\newcommand{\maximumRemainingAgentsSet}{{\mathcal{M}}}
\newcommand{\stepIndex}{{t}}
\newcommand{\emptySequence}{{\emptyset}}
\newcommand{\activeAgents}{{\mathcal{A}}}
\newcommand{\protestingAgents}{{\mathcal{P}}}
\newcommand{\raisingAgents}{{\mathcal{R}}}
\newcommand{\append}{{{}^\frown}}
\newcommand{\coordinatorTerminationCount}{{T}}
\newcommand{\cardinality}[1]{{\lvert{#1}\rvert}}
\newcommand{\Probability}{{\mathbb{P}}}
\newcommand{\goodStateEvent}{{\mathbb{G}}}

\newcommand{\goodTermination}{{\mathbb{T}}}
\newcommand{\badTermination}{{\tilde{\mathbb{T}}}}
\newcommand{\probabilitySpace}{{\Omega}}
\newcommand{\descriptionLength}{{d}}
\newcommand{\currentLength}{{\ell}}

\newcommand{\terminationThreshold}{{\tau}}
\newcommand{\agentInputSet}{{\mathcal{S}}}
\newcommand{\networkGraph}{{\mathcal{G}}}
\newcommand{\numCoordinators}{{c}}
\newcommand{\coordinatorIndex}{{\ell}}
\newcommand{\indicator}[1]{{\mathbf{1}_{{#1}}}}
\newcommand{\compatible}{||}
\newcommand{\generalindex}{{k}}
\newcommand{\generalbit}{{b}}
\newcommand{\residualProbability}{\varepsilon}
\newcommand{\quantizationPrecision}{{p}}

\newtheorem{theorem}{Theorem}

\newtheorem{lemma}{Lemma}

\newtheorem{Prob}{Problem}
\newtheorem{Assume}{Assumption}
\newtheorem{remark}{Remark}
\newtheorem{definition}{Definition}

\title{A Scalable Max-Consensus Protocol For Noisy Ultra-Dense Networks}
\author{
\IEEEauthorblockN{
Navneet Agrawal\IEEEauthorrefmark{1}\IEEEauthorrefmark{2}, Matthias Frey\IEEEauthorrefmark{1}\IEEEauthorrefmark{2} and S\l awomir~Sta\'{n}czak\IEEEauthorrefmark{2}\IEEEauthorrefmark{3}
\\
\IEEEauthorrefmark{2}Technische Universität Berlin, \IEEEauthorrefmark{3}Fraunhofer Heinrich Hertz Institute
}
}

\DefCurrentAudience{arxiv}

\begin{document}
\maketitle

\blankfootnote{
This work was supported by the German Research Foundation (DFG) within their priority program SPP 1914 ``Cyber-Physical Networking'' and by the German Federal Ministry of Education and Research under 
grant 16KIS0605.
\\
\IEEEauthorrefmark{1} The first two authors contributed equally to this work.
}

\begin{abstract}
We introduce \emph{ScalableMax}, a novel communication scheme for achieving max-consensus in a network of multiple agents 
which harnesses the interference in the wireless channel as well as its multicast capabilities.
In a sufficiently dense network, the amount of communication resources required grows logarithmically with the number of nodes, while in state-of-the-art approaches, this growth is at least linear. ScalableMax can handle additive noise and works well in a high SNR regime. For medium and low SNR, we propose the \emph{ScalableMax-EC} scheme, which extends the ideas of ScalableMax by introducing a novel error correction scheme. It achieves lower error rates at the cost of using more channel resources. However, it preserves the logarithmic growth with the number of agents in the system.
\end{abstract}

\section{Introduction and Prior Work}
\label{section:introduction}
The problem of achieving max-consensus in a network of agents arises in many current and envisioned practical applications, particularly in regard to distributed and cooperative control. 
Examples most notably include task assignment \cite{Wongpiromsarn2010}, leader election \cite{Borsche2010}, rendezvous \cite{Sorensen2006}, clock synchronization \cite{Maggs2012}, spectrum sensing \cite{Li2010}, distributed decision-making \cite{Olfati-Saber2004} and formation control \cite{Lafferriere2005}.
Future generations of mobile networks are anticipated to be several orders of magnitude denser than today because of expected infrastructure densification \cite{Ge2016}.
Distributed and cooperative control of multiple agents in various ultra-dense networks will be a major challenge. 
Therefore the growth in complexity of consensus algorithms with the number of agents in the network could become much more important than it is today.
In this work, we present max-consensus protocols that are practical to implement in wireless communication systems and exhibit a more favorable asymptotic complexity behavior than state-of-the-art alternatives.

Historically, max-consensus algorithms are analyzed based on the properties of the communication network graph
\cite{Fax2004,Olfati-Saber2004,Cortes2008}.
The exchange of information between all neighboring agents is assumed to happen simultaneously, with complexity independent of the number of agents.
Hence, these algorithms are designed to minimize the total number of information exchanges required to reach consensus.
However, in wireless networks, these assumptions are often unrealistic due to the presence of interference and noise. 
On the other hand, the specific characteristics of the wireless channel can be exploited by making use of its broadcast and superposition properties. 
Iutzeler et al. proposed and analyzed three communication strategies: 
\emph{Random-Pairwise, Random-Walk} and \emph{Random-Broadcast} \cite{Iutzeler2012, iutzeler2013distributed}. They leverage the broadcast property of the wireless channel,
reduce interference using random scheduling of agents sharing the same channel, 
and protect the transmitted messages using forward error correction. 
This leads to a linear growth of communication resources necessary with the number of agents.
Alternatively, the maximum can be approximated with linear functions
and thus, linear consensus protocols can be applied, e.g., \cite{Nejad2009, Tahbaz-Salehi2006, Cortes2008}.
In \cite{goldenbaum2013robust, Molinari2018, Molinari2018a}, 
the superposition property of the wireless channel is harnessed to achieve 
constant complexity in the number of agents 
in networks with bounded diameter.
But these works neither consider noise introduced by the approximation of the maximum function nor by the wireless channel.
\cite{Huang2010} proposes to use a stochastic approximation based algorithm 
to tackle the residual additive noise, but the convergence rate is much slower than that of standard consensus algorithms.

The main contributions of this work are (a) the introduction of a novel max-consensus protocol that harnesses interference to achieve logarithmic cost while dealing with additive noise, and (b) an error correction mechanism which improves performance in the low and medium SNR regime.

\section{Notation}
We denote the sets of finite and infinite binary sequences with $\finiteBinarySequences$ and $\infiniteBinarySequences$, respectively. Given $\generalBinarySequence_1, \generalBinarySequence_2 \in \finiteBinarySequences \cup \infiniteBinarySequences$, they are \emph{compatible}, or $\generalBinarySequence_1 \compatible \generalBinarySequence_2$, if they coincide on the intersection of their domains. $\generalBinarySequence_1$ is \emph{lexicographically greater} than $\generalBinarySequence_2$, or $\generalBinarySequence_1 \lexGreater \generalBinarySequence_2$, if there is $\generalindex$ such that $\generalBinarySequence_1(\generalindex) > \generalBinarySequence_2(\generalindex)$, while for all $\generalindex' < \generalindex$, $\generalBinarySequence_1(\generalindex') = \generalBinarySequence_2(\generalindex')$. We write $\generalBinarySequence_1 \lexGreaterCompatible \generalBinarySequence_2$ if $\generalBinarySequence_1 \compatible \generalBinarySequence_2$ or $\generalBinarySequence_1(\generalindex) > \generalBinarySequence_2(\generalindex)$. $\emptySequence$ denotes the empty sequence. Given $\generalBinarySequence \in \finiteBinarySequences$ and $\generalbit \in \{0,1\}$, $\generalBinarySequence\append\generalbit$ is the sequence that results from appending $\generalbit$ at the end of $\generalBinarySequence$. Finally, $\indicator{}$ denotes the indicator function and $\lvert \cdot \rvert$ the cardinality of a set.

\section{System Model and Problem Statement}
\label{section:model}
\subsection{Preliminaries}
\label{section:graph_rep}
A Wireless Multiple Access Channel (WMAC) is a system with inputs $\agentTx{1}, \dots, \agentTx{\numAgents}$ and output
$\coordinatorRx := \sum\nolimits_{k=1}^\numAgents \channelFading{k} \agentTx{k} + \noise,$
where $\agentTx{k} \in \complex$ is the signal transmitted by transmitter $k$,
the complex random variable $\channelFading{k}$ is the channel fading coefficient of transmitter $k$,
and the complex random variable $\noise$ is the additive noise at the receiver.
A multicast channel takes an input $\coordinatorTx \in \complex$ from a single transmitting node, and produces outputs $\agentRx{1}, \dots, \agentRx{\numAgents}$ defined as
$\agentRx{k} = \channelFading{k}\coordinatorTx + \agentNoise{k}$ for $\agentIndex=1,\dots,\numAgents$,
where  the complex random variables $\agentNoise{k}$ and $\channelFading{k}$ represent the additive noise and the channel fading coefficient at receiver $k$, respectively. 
One \emph{channel use} is defined as a realization of either a WMAC, a multicast or a point-to-point channel.

\begin{Assume}[Channel Assumptions]
\label{as:channel}
	We assume that the inputs and outputs of the WMACs are real.
	Moreover, the fading coefficients $\channelFading{1}, \dots, \channelFading{\numAgents}$ are assumed to be deterministically equal to $1$.
	The only assumption on the additive noise distribution is that it is symmetric around $0$. White Gaussian noise is one example of such a noise distribution.
\end{Assume}

One way of accommodating complex fading would be to add suitable pre- and post-processing which cancels the fading coefficients up to a residual noise term. For an example of how this can be done even in case the fading coefficients are known neither at the transmitter nor at the receiver, we refer the reader to \cite{Goldenbaum2013a} and \cite{goldenbaum2013robust}.
Multicast communication satisfying arbitrarily low errors can be realized employing state-of-the-art coding schemes with forward error correction. 
In the following, we assume that multicast transmission of binary sequences is possible without error.
Consider a wireless network defined by an undirected connected graph $\networkGraph = (\agentSet, \edgeSet)$. 
The nodes in $\agentSet$ can communicate with each other through channels represented by the edges in $\edgeSet$. 
Besides point-to-point communication along individual edges, we also harness the multicast and superposition (WMAC) properties of the wireless channel.
For simplicity, we start with considering a star-shaped network topology, i.e., there is a central node $\coordinator \in \agentSet$, the \emph{coordinator}, with links to all other nodes in $\agentSet$. 
We assume that the communication between $\coordinator$ and the other nodes is perfectly synchronized.
In Section~\ref{section:noncomplete}, we extend the proposed solutions to general undirected connected wireless network graphs assuming some prior coordination.

\subsection{Problem Statement}
In this section, we define the general max-consensus problem and simplify it to a relaxed version which can be solved more efficiently.

\begin{Prob}[Max-consensus]
\label{Prob:max-consensus}
Each agent $\agent{k} \in \agentSet$ holds an input $\agentInputSequence{k} \in \agentInputSet$,
where $\agentInputSet$ is a finite totally ordered set.
We say that the system has achieved \emph{max-consensus} 
if all agents agree on a common output $\coordinatorOutputEstimate$
that is equal to the maximum of the inputs from all agents,
i.e., $\coordinatorOutputEstimate = \max_{\agent{k} \in \agentSet} \agentInputSequence{k}$.
The objective is to design protocols that can achieve max-consensus with a minimum number of channel uses.
\end{Prob}

We can assume without loss of generality that $\agentInputSet$ is a set of binary sequences of a certain fixed length, equipped with lexicographic ordering which coincides with the usual ordering on dyadic rationals.
For example, consider a Wireless Sensor Network where sensor nodes are sensing a physical phenomenon described by a real number.
The sensors, due to their limited sensitivity, can only read the value up to a quantized number, represented by a finite sequence of binary digits.
In the following, we assume that each agent holds an infinite-length binary sequence, and that no two agents hold the same sequence.
In practice, this can be achieved by concatenating as many uniformly random bits as required by the scheme at the end of any agent's finite input sequence.
In the relaxed version of the max-consensus problem, we seek to narrow down the set of all agents to a smaller set which still contains the agent holding the maximum input. 

\begin{definition}[Weak $\maximumRemainingAgents$-max-consensus]	
	Each agent in $\agentSet$ holds an input sequence $\agentInputSequence{k} \in \infiniteBinarySequences$, where no two inputs are the same.
	At any point in time, the coordinator can terminate the scheme with a termination condition $\outputCondition = \outputCondition(\outputConditionFreeVariable)$ either of the form $\outputCondition(\outputConditionFreeVariable) = \outputConditionFreeVariable \lexGreaterCompatible \coordinatorOutputEstimate$ or of the form $\outputCondition(\outputConditionFreeVariable) = \outputConditionFreeVariable \lexGreater \coordinatorOutputEstimate$, where $\outputConditionFreeVariable$ is a free variable and $\coordinatorOutputEstimate \in \finiteBinarySequences$ is called the coordinator's output estimate. 
	We say that the termination is successful iff $ 1 \leq \lvert \maximumRemainingAgentsSet \rvert \leq \maximumRemainingAgents$, where $ \maximumRemainingAgentsSet := \{ \agent{k}: \outputCondition(\agentInputSequence{k}) \} $ is the set of agents satisfying the termination condition.
\end{definition}

\begin{remark}[From weak $\maximumRemainingAgents$-max-consensus to max-consensus]\label{remark:weak-to-full}
	Further steps are required after reaching a weak $\maximumRemainingAgents$-max-consensus to find the true maximum among the remaining set $\maximumRemainingAgentsSet$ of agents.
	As long as $\maximumRemainingAgents$ does not grow with the number of agents in the system, this reduction 
	can be achieved with a constant number of channel uses through a series of point-to-point and/or multicast communications, e.g., employing Random-Pairwise or Random-Broadcast~\cite{Iutzeler2012, iutzeler2013distributed}.
\end{remark}

\begin{remark}[Designing $\maximumRemainingAgents$]
	The agent holding the true maximum input sequence is guaranteed to be an element of $\maximumRemainingAgentsSet$ as long as $\maximumRemainingAgentsSet \neq \emptyset$.
	$\maximumRemainingAgents$ is a designable parameter which does not need to grow with the number of agents in the system.
	The higher it is, the more we can harness the combined signal strength of multiple transmitters to combat noise, 
	but the more communication resources are necessary to simplify the max-consensus problem to weak $\maximumRemainingAgents$-max-consensus.
\end{remark}

\section{ScalableMax Scheme}
\label{section:scalable-max}
In this section, we propose a scheme that achieves weak $\maximumRemainingAgents$-max-consensus and scales logarithmically with the number of agents.
The max-consensus problem can be simplified to weak $\maximumRemainingAgents$-max-consensus as pointed out in Remark~\ref{remark:weak-to-full}.

The coordinator starts the scheme and generates an output estimate $\coordinatorOutputEstimate \in \finiteBinarySequences$ based on information received from the agents.
In the following, we detail the communication protocols and information shared between agents and the coordinator.
For every possible coordinator output estimate $\coordinatorOutputEstimate$, we define the set $\protestingAgents_\coordinatorOutputEstimate := \{\agentIndex: \agentInputSequence{\agentIndex} \lexGreater \coordinatorOutputEstimate\}$ of \emph{protesting agents}, the set $\activeAgents_\coordinatorOutputEstimate :=  \{\agentIndex: \agentInputSequence{\agentIndex} \lexGreaterCompatible \coordinatorOutputEstimate\}$ of \emph{active agents} and the set $\raisingAgents_\coordinatorOutputEstimate :=  \{\agentIndex: \agentInputSequence{\agentIndex} \lexGreaterCompatible \coordinatorOutputEstimate\append1\}$ of \emph{raising agents}. The coordinator uses noisy estimates of the cardinalities of these sets in order to refine its output estimate.

We use an iteration counter $\stepIndex$, where each iteration consists of a transmission of digital information through the multicast channel and three uses of the WMAC, and thereby corresponds to a constant number of channel uses. Conceptually, we thus split every iteration $\stepIndex$ into four time instants $4\stepIndex, \dots, 4\stepIndex+3$. The coordinator starts with $\coordinatorOutputEstimate(0) := \emptySequence$. At every time instant of the form $4\stepIndex$, it transmits $\coordinatorOutputEstimate(\stepIndex)$ through the multicast channel. We remark that since $\coordinatorOutputEstimate(\stepIndex)$ differs from $\coordinatorOutputEstimate(\stepIndex-1)$ in at most one digit, it is sufficient to transmit only the change, and hence, the length of the transmitted sequence can be considered constant. At time instants not divisible by 4, the agents transmit through the WMAC, the signal of each being either $1$ or $0$, according to the following scheme:

\begin{tabular}{ccc}
time instant            & step name      & signal transmitted by $\agent{\agentIndex}$\\
$4\stepIndex + 1$       & protest        & $\indicator{\agentIndex \in \protestingAgents_{\coordinatorOutputEstimate(\stepIndex)}}$\\
$4\stepIndex + 2$       & activity       & $\indicator{\agentIndex \in \activeAgents_{\coordinatorOutputEstimate(\stepIndex)}}$\\
$4\stepIndex + 3$       & raising        & $\indicator{\agentIndex \in \raisingAgents_{\coordinatorOutputEstimate(\stepIndex)}}$
\end{tabular}

We denote the signal transmitted by agent $\agent{\agentIndex}$ at time instant $\stepIndex$ with $\agentTx{\agentIndex}(\stepIndex)$ and the corresponding received signal $\coordinatorRx(\stepIndex) = \sum_{\agentIndex=1}^{\numAgents} \agentTx{\agentIndex}(\stepIndex) + \noise(\stepIndex)$. These values are not defined if $\stepIndex$ is divisible by $4$, since the agents do not transmit in these steps.

After step $4\stepIndex+3$, the coordinator either determines a new output estimate $\coordinatorOutputEstimate(\stepIndex+1)$ or it makes a termination decision according to Fig.~\ref{algorithm:coordinator-postprocessing}. \showto{arxiv}{In Fig.~\ref{fig:coordinator-postprocessing}, we show a graphical representation of part of the decision process.} \showto{conference}{The extended version of this work~\cite{arxivVersion} includes a further graphical representation to clarify the decision process.}

\begin{algorithm}
\If{$\coordinatorRx(4\stepIndex+1) > \maximumRemainingAgents/4$}{
  Terminate with $\outputCondition(\outputConditionFreeVariable) = \outputConditionFreeVariable \lexGreater \coordinatorOutputEstimate(\stepIndex)$\;
}
\If{$\coordinatorRx(4\stepIndex+2) < 3\maximumRemainingAgents/4$}{
  Terminate with $\outputCondition(\outputConditionFreeVariable) = \outputConditionFreeVariable \lexGreaterCompatible \coordinatorOutputEstimate(\stepIndex)$\;
}
\eIf{$\coordinatorRx(4\stepIndex+3) < \maximumRemainingAgents/4$}{
  $\coordinatorOutputEstimate(\stepIndex+1) \gets \coordinatorOutputEstimate(\stepIndex)\append0$\;
  }{
  $\coordinatorOutputEstimate(\stepIndex+1) \gets \coordinatorOutputEstimate(\stepIndex)\append1$\;
  \If{$\coordinatorRx(4\stepIndex+3) < 3\maximumRemainingAgents/4$}{
    Terminate with $\outputCondition(\outputConditionFreeVariable) = \outputConditionFreeVariable \lexGreaterCompatible \coordinatorOutputEstimate(\stepIndex+1)$\;
  }
}
\caption{Post-processing of received signals in ScalableMax.}
\label{algorithm:coordinator-postprocessing}
\end{algorithm}

\begin{shownto}{arxiv}
\begin{figure}
\begin{tikzpicture}
\coordinate  (origin)                    at (0,0);
\coordinate  (noprotest_mactivity)       at (0,4);
\coordinate  (mprotest)                  at (4,4);
\coordinate  (noprotest_upperedge)       at (0,5);
\coordinate  (mprotest_upperedge)        at (5,5);
\coordinate  (atermination_leftedge)     at (-1,3);
\coordinate  (atermination_ptermination) at (1,3);
\coordinate  (ptermination_loweredge)    at (1,-1);
\coordinate  (ptermination_upperedge)    at (1,5);
\coordinate  (midprotest_midactivity)    at (2,2);
\coordinate  (ptermination_correction_upperedge) at (3,5);
\coordinate  (ptermination_correction_loweredge) at (3,-1);
\coordinate  (atermination_correction) at (1,1);
\coordinate  (atermination_correction_leftedge) at (-1,1);

\fill[pattern=north east lines, pattern color=lightgray] (midprotest_midactivity) rectangle (noprotest_upperedge);

\node[rotate=90] at (-1.5,2) {$\coordinatorRx(4\stepIndex+2)$ (activity step)};
\node[below] at (2,-1) {$\coordinatorRx(4\stepIndex+1)$ (protest step)};
\node[anchor=north] at (origin) {$(0,0)$};
\draw[thick] (mprotest) -- +(-4pt,0) -- +(4pt,0) node[right] (mprotest_label) {$(\maximumRemainingAgents,\maximumRemainingAgents)$};
\draw[thick] (noprotest_mactivity) -- +(4pt,0) -- +(-4pt,0) node[left] {$(0,\maximumRemainingAgents)$};
\node[above left=5pt,circle,draw] at (atermination_ptermination) {1};
\node[below left=5pt,circle,draw] at (atermination_ptermination) {2};
\node[right=5pt,circle,draw] at (atermination_ptermination) {3};
\node[below left=5pt,circle,draw] at (atermination_correction) {4};
\node[left=5pt,circle,draw] at (mprotest) {4};

\draw[dashed] (origin) -- (noprotest_upperedge);
\draw[dashed] (origin) -- (mprotest_upperedge);

\draw (atermination_leftedge) -- (atermination_ptermination);
\draw (ptermination_loweredge) -- (ptermination_upperedge);
\draw (atermination_correction_leftedge) -- (atermination_correction);
\draw (ptermination_correction_upperedge) -- (ptermination_correction_loweredge);
\end{tikzpicture}
\caption{Visualization of post-processing at the coordinator of received signals $\coordinatorRx(4\stepIndex+1)$ and $\coordinatorRx(4\stepIndex+2)$. Dashed lines show the possible noiseless combined signals and solid lines delimit the numbered decision regions. Only in decision region 1 is $\coordinatorRx(4\stepIndex+3)$ taken into account, and the decision can be to append $0$ or $1$ to the output estimate, or to append $1$ and terminate with a $\lexGreaterCompatible$ condition. Received signals in region 2 and 3 lead to termination with conditions $\lexGreaterCompatible$ and $\lexGreater$, respectively, with the current unmodified output estimate. In ScalableMax, the regions marked 4 are part of regions 2 and 3, respectively. In ScalableMax-EC, they correspond to the coordinator removing the last digit from its current output estimate, thus correcting an error that may have been made in previous steps due to high noise.}
\label{fig:coordinator-postprocessing}
\end{figure}
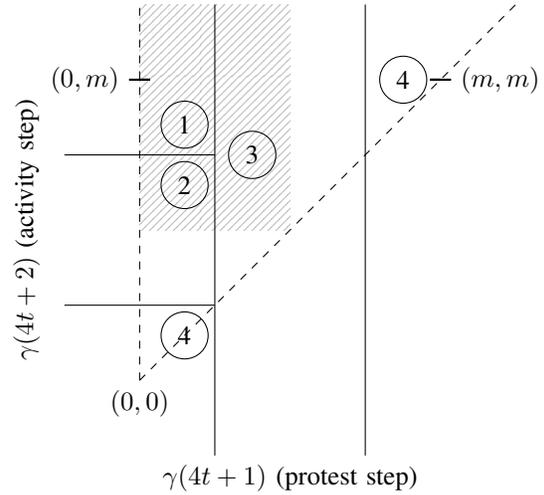
\end{shownto}

\begin{theorem}
\label{theorem:ScalableMax}
Suppose that $\maximumRemainingAgents$ is even. Then the probability that the ScalableMax scheme terminates successfully within $\descriptionLength+1$ iterations is at least $\Probability(\noise \leq \maximumRemainingAgents/4)^{3(\descriptionLength+1)}$, where $\noise$ is the additive noise of the WMAC and
\[\descriptionLength := \min\{ \ell \in \naturals: \forall \coordinatorOutputEstimate \in \{0,1\}^\ell~ \cardinality{\{\agentIndex: \agentInputSequence{\agentIndex} \compatible \coordinatorOutputEstimate\}} \leq 1\}.\]
\end{theorem}

\begin{remark}
The maximum description length $\descriptionLength$ is an important quantity for the performance of the ScalableMax scheme, but it is highly dependent on the agents' inputs and therefore in general unknown. However, we can bound it under assumptions on the distribution of the agents' inputs. So for instance, if we assume the inputs are uniformly distributed, we get
\begin{align*}
\Probability(\descriptionLength \geq \descriptionLength_0)
&=
\Probability\left(
  \bigvee\nolimits_{\agentIndex \neq \agentIndex' = 1}^{\numAgents}
    \bigwedge\nolimits_{\currentLength=1}^{\descriptionLength_0} \agentInputSequence{\agentIndex}(\currentLength) = \agentInputSequence{\agentIndex'}(\currentLength)
\right)
\\
&\leq
\numAgents(\numAgents-1)/2 \cdot (1/2)^{\descriptionLength_0},
\end{align*}
where the inequality is due to the union bound. Elementary transformations yield that $\Probability(\descriptionLength \geq \descriptionLength_0) \leq \residualProbability$ as long as
\[ \descriptionLength_0 \geq \log_2 \numAgents + \log_2 (\numAgents - 1) + \log_2(1/\residualProbability) - 1, \]
so in the case that the agents' inputs are uniformly distributed, the description length depends logarithmically on $\numAgents$.

In general, according to our assumptions in Section~\ref{section:model}, each of the agents' inputs consists of finitely many, say $\quantizationPrecision$, arbitrary bits and infinitely many uniform bits. Then $\Probability(\descriptionLength \geq \descriptionLength_0) \leq \residualProbability$ if
\[ \descriptionLength_0 \geq \quantizationPrecision + \log_2 \numAgents + \log_2 (\numAgents - 1) + \log_2(1/\residualProbability) - 1.\]
\end{remark}

\begin{shownto}{arxiv}
In order to prove Theorem~\ref{theorem:ScalableMax}, we introduce some additional notation and terminology. We say that the system is in a \emph{good state} if the output estimate $\coordinatorOutputEstimate$ satisfies $\cardinality{\protestingAgents_\coordinatorOutputEstimate} < \maximumRemainingAgents/2$ and $\cardinality{\activeAgents_\coordinatorOutputEstimate} \geq \maximumRemainingAgents/2$. The region of potential signals to be received in the protest and activity steps provided the system is in a good state is shaded in Fig.~\ref{fig:coordinator-postprocessing}. Note that if $\coordinatorOutputEstimate$ is the empty sequence, the state is good if there are at least $\maximumRemainingAgents/2$ agents in the system. With $\goodStateEvent_\stepIndex$ we denote the event that the system is in a good state at step $4\stepIndex$. With $\goodTermination_\stepIndex$ we denote the event that the coordinator terminates the scheme successfully at step $4\stepIndex$, with $\badTermination_\stepIndex$ we denote the event that the coordinator terminates the scheme unsuccessfully at step $4\stepIndex$ and with $\badTermination_\stepIndex^\outputCondition$ we denote the event that the coordinator terminates the scheme unsuccessfully at step $4\stepIndex$ with output condition $\outputCondition$.

\begin{lemma}
\label{lemma:nocorrection-transitionprobability}
\[
\Probability(\goodStateEvent_{\stepIndex+1} \cup \goodTermination_{\stepIndex+1} | \goodStateEvent_{\stepIndex}) \geq \Probability(\noise \leq \maximumRemainingAgents/4)^3.
\]
\end{lemma}
\begin{proof}
We abbreviate the three possible termination conditions at step $4\stepIndex$ with $\outputCondition_1(\outputConditionFreeVariable) = \outputConditionFreeVariable \lexGreater \coordinatorOutputEstimate(\stepIndex)$, $\outputCondition_2(\outputConditionFreeVariable) = \outputConditionFreeVariable \lexGreaterCompatible \coordinatorOutputEstimate(\stepIndex)$ and $\outputCondition_3(\outputConditionFreeVariable) = \outputConditionFreeVariable \lexGreaterCompatible \coordinatorOutputEstimate(\stepIndex+1)$ and bound the probability as follows:
\begin{align}
\nonumber
&\begin{aligned}\hphantom{={}}
\Probability(\goodStateEvent_{\stepIndex+1} \cup \goodTermination_{\stepIndex+1} | \goodStateEvent_{\stepIndex})
\end{aligned}
\\
\nonumber
&\begin{aligned}\geq
\Probability((\goodStateEvent_{\stepIndex+1} \cup \goodTermination_{\stepIndex+1}) \cap \goodStateEvent_{\stepIndex})
\end{aligned}
\\
\label{eq:nocorrection-transitionprobability-splitting}
&\begin{aligned}
=
1-\Probability(
  &(\badTermination_{\stepIndex+1}^{\outputCondition_1} \cap \goodStateEvent_{\stepIndex})
  \cup
  (\badTermination_{\stepIndex+1}^{\outputCondition_2} \cap \goodStateEvent_{\stepIndex})
  \cup
  (\badTermination_{\stepIndex+1}^{\outputCondition_3} \cap \goodStateEvent_{\stepIndex})
  \\
  &\cup
  (\goodStateEvent_{\stepIndex} \setminus (\goodStateEvent_{\stepIndex+1} \cup \goodTermination_{\stepIndex+1} \cup \badTermination_{\stepIndex+1}))
)
\end{aligned}
\\
\label{eq:nocorrection-transitionprobability-implications}
&\begin{aligned}\geq
1-\Probability(
  &\noise(4\stepIndex+2) > \maximumRemainingAgents/4
  \cup
  \noise(4\stepIndex+3) < -\maximumRemainingAgents/4
\\
  &\cup
  \noise(4\stepIndex+4) > \maximumRemainingAgents/4
)
\end{aligned}
\\
\nonumber
&\begin{aligned}=
\Probability(
  &\noise(4\stepIndex+2) \leq \maximumRemainingAgents/4
  \cap
  \noise(4\stepIndex+3) \geq -\maximumRemainingAgents/4
\\
  &\cap
  \noise(4\stepIndex+4) \leq \maximumRemainingAgents/4
)
\end{aligned}
\\
\label{eq:nocorrection-transitionprobability-independence}
&\begin{aligned}=
\Probability(\noise \leq \maximumRemainingAgents/4)^3
\end{aligned}
\end{align}
Equality (\ref{eq:nocorrection-transitionprobability-splitting}) is due to the observation that if neither $\goodStateEvent_{\stepIndex+1}$ nor $\goodTermination_{\stepIndex+1}$ occurs, at least one of the following events must occur: $\badTermination_{\stepIndex+1}^{\outputCondition_1}$, $\badTermination_{\stepIndex+1}^{\outputCondition_2}$, $\badTermination_{\stepIndex+1}^{\outputCondition_3}$, $\probabilitySpace \setminus (\goodStateEvent_{\stepIndex+1} \cup \goodTermination_{\stepIndex+1} \cup \badTermination_{\stepIndex+1})$.

Inequality (\ref{eq:nocorrection-transitionprobability-implications}) is due to the following implications among events:
\begin{itemize}
 \item $\goodStateEvent_{\stepIndex} \cap \badTermination_{\stepIndex+1}^{\outputCondition_1}$ implies $\protestingAgents_{\coordinatorOutputEstimate(\stepIndex)} = \emptyset$ and therefore $\noise(4\stepIndex+2) > \maximumRemainingAgents/4$.
 \item $\goodStateEvent_{\stepIndex} \cap \badTermination_{\stepIndex+1}^{\outputCondition_2}$ implies $\activeAgents_{\coordinatorOutputEstimate(\stepIndex)} > \maximumRemainingAgents$ and therefore $\noise(4\stepIndex+3) < -\maximumRemainingAgents/4$.
 \item In order to analyze the implications of $(\goodStateEvent_{\stepIndex} \cap \badTermination_{\stepIndex+1}^{\outputCondition_3}) \cup (\goodStateEvent_{\stepIndex} \setminus (\goodStateEvent_{\stepIndex+1} \cup \goodTermination_{\stepIndex+1} \cup \badTermination_{\stepIndex+1}))$, we distinguish three cases. First, we consider $\raisingAgents_{\coordinatorOutputEstimate(\stepIndex)} = \emptyset$. In this case the event implies $\noise(4\stepIndex+4) > \maximumRemainingAgents/4$. Second, we consider $0 < \cardinality{\raisingAgents_{\coordinatorOutputEstimate(\stepIndex)}} \leq \maximumRemainingAgents/2$. In this case, the event implies that $1$ was appended to $\coordinatorOutputEstimate$ but no termination occurred and thus $\noise(4\stepIndex+4) > \maximumRemainingAgents/4$. Finally, we consider $\cardinality{\raisingAgents_{\coordinatorOutputEstimate(\stepIndex)}} > \maximumRemainingAgents$. In this case, the event implies $\noise(4\stepIndex+4) < -\maximumRemainingAgents/4$. Due to the symmetry of the noise, we can assume that in all three cases the event implies $\noise(4\stepIndex+4) > \maximumRemainingAgents/4$ for the sake of probability statements.
\end{itemize}

Finally, equality (\ref{eq:nocorrection-transitionprobability-independence}) is due to the independence of $\noise(4\stepIndex+2)$, $\noise(4\stepIndex+3)$ and $\noise(4\stepIndex+4)$.
\end{proof}

\begin{proof}[Proof of Theorem~\ref{theorem:ScalableMax}]
We observe that after $4(\descriptionLength+1)$ steps, the scheme either terminates or it is in a bad state and that by definition the scheme always starts in a good state (i.e. $\goodStateEvent_0$ almost surely occurs), so we can bound the probability of successful termination after at most $4(\descriptionLength+1)$ steps as
\begin{align*}
\Probability\left(
  \bigcap\limits_{\stepIndex=1}^{\descriptionLength+1}\left(
    \goodStateEvent_\stepIndex
    \cup
    \bigcup\limits_{\stepIndex'=1}^\stepIndex \goodTermination_{\stepIndex'}
  \right)
\right)
&=
\prod\limits_{\stepIndex=1}^{\descriptionLength+1} \Probability(\goodStateEvent_\stepIndex \cup \goodTermination_\stepIndex | \goodStateEvent_{\stepIndex-1})
\\
&\geq
\Probability(\noise \leq \maximumRemainingAgents/4)^{3(\descriptionLength+1)},
\end{align*}
where the inequality is due to Lemma~\ref{lemma:nocorrection-transitionprobability}, observing that conditioned under the previous state of the system, the next step of the scheme is independent of the previous steps.
\end{proof}
\end{shownto}
\begin{shownto}{conference}
We omit the proof of the theorem due to lack of space and instead present only a brief sketch. For a full formal proof, we refer the reader to the extended version of this paper~\cite{arxivVersion}.
With an appropriate case distinction, it is not hard to show that the scheme succeeds if the noise realizations throughout the scheme never exceed $\maximumRemainingAgents/4$. We can therefore derive an error bound if we know through how many iterations the scheme has to go until termination. Observe that the coordinator's output estimate cannot reach a length of more than $\descriptionLength$ if the noise is always less than $\maximumRemainingAgents/4$. Moreover, each digit added to the coordinator's output estimate as well as the termination decision in the end corresponds to one iteration and thus entails three analog multicast steps. For these reasons, we can argue that in case the noise samples are always less than $\maximumRemainingAgents/4$, the scheme terminates after at most $\descriptionLength+1$ iterations and we observe at most $3(\descriptionLength+1)$ noise samples, from which the theorem follows.
\end{shownto}

\section{ScalableMax-EC Scheme}
\label{section:scalable-max-ec}
In this section, we introduce the ScalableMax-EC scheme which expands upon the ideas of the previous section, introducing error correction. We achieve this with two main modifications to the ScalableMax scheme. First, the coordinator can now additionally make correction decisions, i.e., remove a digit from its current output estimate $\coordinatorOutputEstimate(\stepIndex)$. Second, in the cases in which the above scheme would terminate, the coordinator does not do so immediately, but rather raises a termination counter and only terminates when this counter reaches a termination threshold $\terminationThreshold$, which is a parameter of the scheme. For each condition $\text{cond} \in \{``\lexGreater", ``\lexGreaterCompatible", ``\text{append}"\}$ and each possible output estimate $\coordinatorOutputEstimate \in \finiteBinarySequences$, the coordinator keeps a termination counter $\coordinatorTerminationCount(\coordinatorOutputEstimate, \text{cond})$, which is initially $0$. Coordinator and agents communicate as they do in the above scheme, but the post-processing in the coordinator after step $4\stepIndex+3$ differs and is conducted according to Fig.~\ref{algorithm:coordinator-postprocessing-correction}. \showto{arxiv}{We visualize a part of this decision process in Fig.~\ref{fig:coordinator-postprocessing}.} \showto{conference}{The extended version~\cite{arxivVersion} of this work includes an additional visualization of part of this decision process.}

\begin{algorithm}
\uIf{$\coordinatorRx(4\stepIndex+1) > 3\maximumRemainingAgents/4$}{
  $\coordinatorOutputEstimate(\stepIndex+1) \gets \coordinatorOutputEstimate(\stepIndex)$ with last digit removed (if any)\;
}
\uElseIf{$\coordinatorRx(4\stepIndex+1) > \maximumRemainingAgents/4$}{
  $\coordinatorTerminationCount(\coordinatorOutputEstimate(\stepIndex), ``\lexGreater") \gets \coordinatorTerminationCount(\coordinatorOutputEstimate(\stepIndex), ``\lexGreater") + 1$\;
  \If{$\coordinatorTerminationCount(\coordinatorOutputEstimate(\stepIndex), ``\lexGreater") = \terminationThreshold$}{
    Terminate with $\outputCondition(\outputConditionFreeVariable) = \outputConditionFreeVariable \lexGreater \coordinatorOutputEstimate(\stepIndex)$\; 
  }
}
\uElseIf{$\coordinatorRx(4\stepIndex+2) < \maximumRemainingAgents/4$}{
  $\coordinatorOutputEstimate(\stepIndex+1) \gets \coordinatorOutputEstimate(\stepIndex)$ with last digit removed (if any)\;
}
\uElseIf{$\coordinatorRx(4\stepIndex+2) < 3\maximumRemainingAgents/4$}{
  $\coordinatorTerminationCount(\coordinatorOutputEstimate(\stepIndex), ``\lexGreaterCompatible") \gets \coordinatorTerminationCount(\coordinatorOutputEstimate(\stepIndex), ``\lexGreaterCompatible") + 1$\;
  \If{$\coordinatorTerminationCount(\coordinatorOutputEstimate(\stepIndex), ``\lexGreaterCompatible") = \terminationThreshold$}{
    Terminate with $\outputCondition(\outputConditionFreeVariable) = \outputConditionFreeVariable \lexGreaterCompatible \coordinatorOutputEstimate(\stepIndex)$\; 
  }
}
\uElseIf{$\coordinatorRx(4\stepIndex+3) < \maximumRemainingAgents/4$}{
  $\coordinatorOutputEstimate(\stepIndex + 1) \gets \coordinatorOutputEstimate(\stepIndex) \append 0$\;
}
\uElseIf{$\coordinatorRx(4\stepIndex+3) < 3\maximumRemainingAgents/4$}{
  $\coordinatorTerminationCount(\coordinatorOutputEstimate(\stepIndex), ``\text{append}") \gets \coordinatorTerminationCount(\coordinatorOutputEstimate(\stepIndex), ``\text{append}") + 1$\;
  \If{$\coordinatorTerminationCount(\coordinatorOutputEstimate(\stepIndex), ``\text{append}") = \terminationThreshold$}{
    Terminate with $\outputCondition(\outputConditionFreeVariable) = \outputConditionFreeVariable \lexGreaterCompatible \coordinatorOutputEstimate(\stepIndex) \append 1$\; 
  }
}
\Else{
  $\coordinatorOutputEstimate(\stepIndex + 1) \gets \coordinatorOutputEstimate(\stepIndex) \append 1$ \;
}
\caption{Post-processing of received signals in ScalableMax-EC.}
\label{algorithm:coordinator-postprocessing-correction}
\end{algorithm}

\section{Simulation Results}
\label{section:simulation}
We model the problem as described in Section \ref{section:model}
and run the ScalableMax and ScalableMax-EC algorithms as described in sections \ref{section:scalable-max} and \ref{section:scalable-max-ec}. 
Uniform random bits are used as the agents' input sequences $(\agentInputSequence{\agentIndex})_{\agentIndex=1}^\numAgents$. We assume unit transmission power and white Gaussian noise. Each plotted data point is an average over $10^5$ identical and independent simulation runs.

In Fig.~\ref{fig:1}, we compare the performance of the proposed schemes in terms of \emph{error rate}, i.e., the rate of unsuccessful termination of the scheme, 
and in Fig.~\ref{fig:2}, the average number of iterations required in successfully terminated runs of the schemes.
As the noise power increases, so does the chance of unfavorable decisions by the coordinator. ScalableMax-EC has a mechanism to correct such bad decisions, and thus exhibits lower error rates, but at the cost of needing more iterations than ScalableMax.

\begin{figure}
\begin{center}
\subfloat[Error rate\label{fig:1}]{
\begin{tikzpicture}
\begin{semilogyaxis}[
  xlabel={Noise power in dB},
  ylabel={Error rate},
  ymin=.0005,
  ymax=1,
  xmin=-5.5,
  xmax=15.5,
  legend style={at={(0.5,1.1)}, anchor=south},
  legend columns=4,
]
\addplot [blue, mark=*, discard if not={correction}{False}] table[x=noise_power, y expr=1-\thisrow{success_rate}, col sep=comma]{results_mutiple_correction_m8_1e5iter_7244a4555.csv};
\addplot [red, mark=square*, discard if not={termination_parameter}{2}] table[x=noise_power, y expr=1-\thisrow{success_rate}, col sep=comma]{results_mutiple_correction_m8_1e5iter_7244a4555.csv};
\addplot [green, mark=triangle*, discard if not={termination_parameter}{3}] table[x=noise_power, y expr=1-\thisrow{success_rate}, col sep=comma]{results_mutiple_correction_m8_1e5iter_7244a4555.csv};
\addplot [brown, mark=diamond*, discard if not={termination_parameter}{4}] table[x=noise_power, y expr=1-\thisrow{success_rate}, col sep=comma]{results_mutiple_correction_m8_1e5iter_7244a4555.csv};
\addplot [cyan, mark=x, discard if not={termination_parameter}{5}] table[x=noise_power, y expr=1-\thisrow{success_rate}, col sep=comma]{results_mutiple_correction_m8_1e5iter_7244a4555.csv};
\addplot [teal, mark=+, discard if not={termination_parameter}{10}] table[x=noise_power, y expr=1-\thisrow{success_rate}, col sep=comma]{results_mutiple_correction_m8_1e5iter_7244a4555.csv};
\addplot [violet, mark=star, discard if not={termination_parameter}{20}] table[x=noise_power, y expr=1-\thisrow{success_rate}, col sep=comma]{results_mutiple_correction_m8_1e5iter_7244a4555.csv};
\legend{no EC, $\terminationThreshold=2$, $\terminationThreshold=3$, $\terminationThreshold=4$, $\terminationThreshold=5$, $\terminationThreshold=10$, $\terminationThreshold=20$}
\end{semilogyaxis}
\end{tikzpicture}
} 
\\

\subfloat[Number of iterations until successful termination.\label{fig:2}]{
\begin{tikzpicture}
\begin{semilogyaxis}[
  xlabel={Noise power in dB},
  ylabel={Iterations},
  ymin=7,
  ymax=250,
  xmin=-5.5,
  xmax=15.5,
]
\addplot [blue, solid, mark=*, discard if not={correction}{False}] table[x=noise_power, y=average_iterations_in_successful_runs, col sep=comma]{results_mutiple_correction_m8_1e5iter_7244a4555.csv};
\addplot [red, solid, mark=square*, discard if not={termination_parameter}{2}] table[x=noise_power, y=average_iterations_in_successful_runs, col sep=comma]{results_mutiple_correction_m8_1e5iter_7244a4555.csv};
\addplot [green, solid, mark=triangle*, discard if not={termination_parameter}{3}] table[x=noise_power, y=average_iterations_in_successful_runs, col sep=comma]{results_mutiple_correction_m8_1e5iter_7244a4555.csv};
\addplot [brown, solid, mark=diamond*, discard if not={termination_parameter}{4}] table[x=noise_power, y=average_iterations_in_successful_runs, col sep=comma]{results_mutiple_correction_m8_1e5iter_7244a4555.csv};
\addplot [cyan, solid, mark=x, discard if not={termination_parameter}{5}] table[x=noise_power, y=average_iterations_in_successful_runs, col sep=comma]{results_mutiple_correction_m8_1e5iter_7244a4555.csv};
\addplot [teal, solid, mark=+, discard if not={termination_parameter}{10}] table[x=noise_power, y=average_iterations_in_successful_runs, col sep=comma]{results_mutiple_correction_m8_1e5iter_7244a4555.csv};
\addplot [violet, solid, mark=star, discard if not={termination_parameter}{20}] table[x=noise_power, y=average_iterations_in_successful_runs, col sep=comma]{results_mutiple_correction_m8_1e5iter_7244a4555.csv};
\end{semilogyaxis}
\end{tikzpicture}
} 
\caption{Performance of ScalableMax (no EC) and ScalableMax-EC with various termination thresholds $\terminationThreshold$, $\numAgents=1000$, and $\maximumRemainingAgents=8$.}
\end{center}
\end{figure}
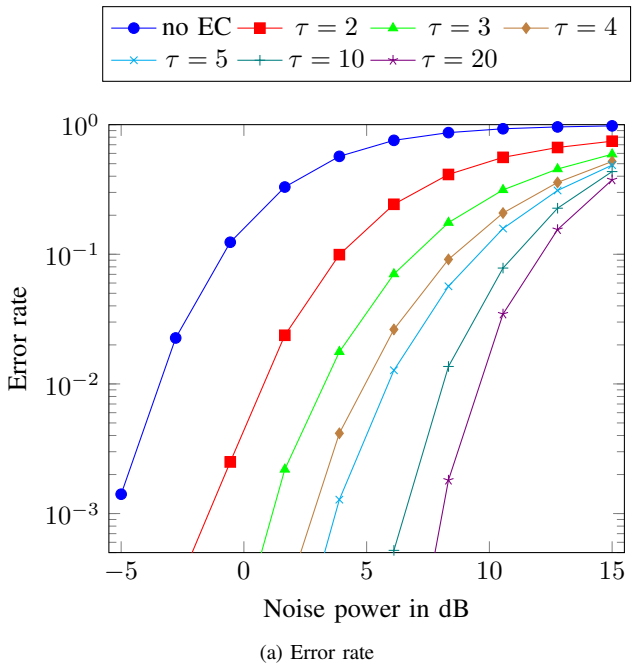
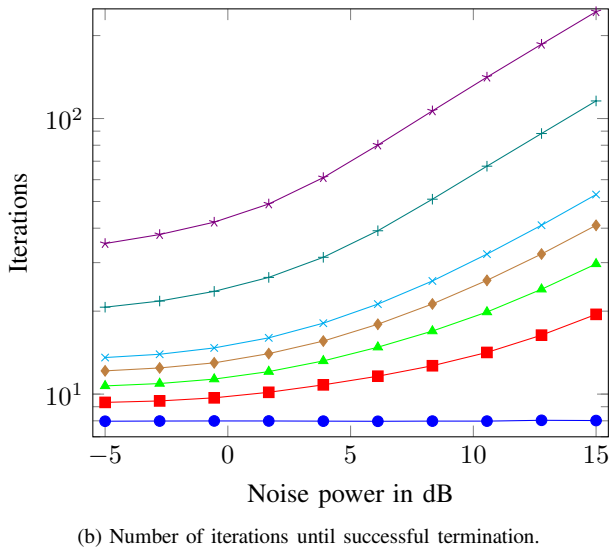

In Fig.~\ref{fig:num_iter}, we compare the scalability of the ScalableMax-EC scheme with the state-of-the-art \emph{Random-Broadcast} (RB) and \emph{Random-Pairwise} (RP) schemes described in \cite{Iutzeler2012}. To this end, we extend our scheme with a RB step to determine the maximum among the $\maximumRemainingAgents=8$ agents that remain after a ScalableMax-EC run. We choose $\terminationThreshold$ such that this combination achieves an overall error rate of at most $0.005$. For comparison, we plot the number of iterations necessary in RB and RP to achieve an error rate of $0.005$, given that all digital transmissions arrive error free.
The ScalableMax-EC scheme scales logarithmically with the total number of agents, 
while RB and RP scale at least linearly (see also~\cite{Iutzeler2012}).

\begin{figure}
\resizebox {\columnwidth} {!} {
\begin{tikzpicture}
\begin{axis}[
  xlabel={$\numAgents$},
  axis y line*=left,
  ylabel={Iterations for ScalableMax-EC},
  ymin=0,
  ymax=200,
  xmin=0,
  xmax=5100,
  legend pos = north west,
]
\addplot [blue,solid,mark=*, discard if not={termination_parameter}{2} ] table[x=agents, y expr=\thisrow{average_iterations_in_successful_runs}+51, col sep=comma] {results_varying_agents_m8_1e5iter_7244a4555.csv};
\addplot [blue,solid,mark=triangle*, discard if not={termination_parameter}{6}] table[x=agents, y expr=\thisrow{average_iterations_in_successful_runs}+51, col sep=comma] {results_varying_agents_m8_1e5iter_7244a4555.csv};
\addplot [blue,solid, mark=square*, discard if not={termination_parameter}{10}] table[x=agents, y expr=\thisrow{average_iterations_in_successful_runs}+51, col sep=comma] {results_varying_agents_m8_1e5iter_7244a4555.csv};
\legend{{-1.0 dB, $\terminationThreshold=2$}, {5.0 dB, $\terminationThreshold=6$}, {7.0 dB, $\terminationThreshold=10$}}
\end{axis}
%
\begin{axis}[
  xlabel={$\numAgents$},
  axis y line*=right,
  ylabel={Iterations for RB and RP},
  ymin=0,
  ymax=40000,
  xmin=0,
  xmax=5100,
  legend pos = south east,
]
\addplot [red, dashed, mark=square*] table[x=agents, y=RB, col sep=comma]{results_varying_agents_RBRP.csv};
\addplot [red, dashed, mark=triangle*] table[x=agents, y=RP, col sep=comma]{results_varying_agents_RBRP.csv};
\legend{{RB}, {RP}}
\end{axis}
\end{tikzpicture}
}
\caption{Number of iterations as a function of $\numAgents$ for proposed ScalableMax-EC scheme on the left y-axis, and state-of-the-art Random Broadcast (RB) and Random Pairwise(RP) schemes on the right y-axis, for a fixed error rate $<0.005$. Note the different scales of left and right y-axes.}
\label{fig:num_iter}
\end{figure}
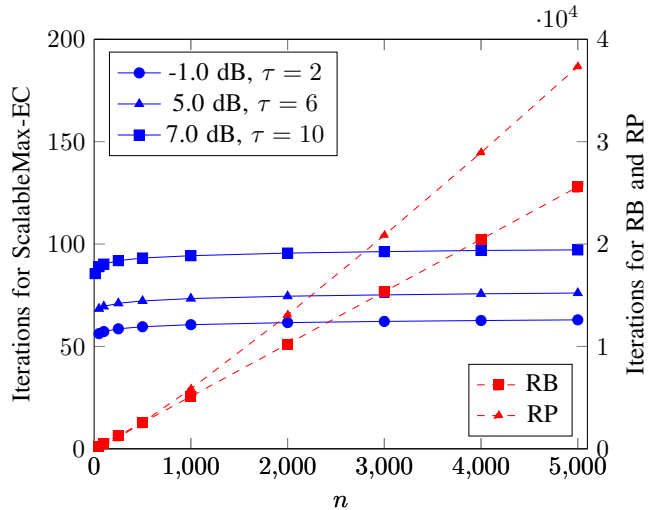

We conclude with two observations made during the simulations which are not shown in the plots. First, if the agents' input sequences are quantized versions of Gaussian random numbers, the number of iterations needed increases slightly depending on the variance of the random numbers and the granularity of the quantization. Second, the performance in terms of iterations can be improved significantly by choosing a suitable $\coordinatorOutputEstimate(0)$ other than $\emptySequence$. One example that performs well is the coordinator output estimate at which the scheme was terminated successfully in an identically distributed but independent earlier simulation run. Finding other ways to choose suitable $\coordinatorOutputEstimate(0)$ in practical scenarios remains an open point for future research.

\section{Extension to non-star-shaped networks}
\label{section:noncomplete}
In this section, we propose a method to extend our schemes to general undirected connected network graphs.
We assume some prior coordination in that a set of designated coordinators $\{\coordinator_1, \dots, \coordinator_\numCoordinators\}$ is known
such that the graph would still be connected if all edges that are not adjacent to one of the coordinators were removed. Furthermore, we assume that these coordinators have a way of scheduling their communication in a sequential manner.
Note that because of the connectivity requirement, 
some agents necessarily have links with two or more coordinators. 
We denote the subgraph of $\networkGraph$ induced by $\coordinator_\coordinatorIndex$ and its neighbors 
with $\networkGraph_\coordinatorIndex$ and achieve max-consensus with these steps:
\begin{enumerate}
 \item \label{extended-network-scheme-step} For each $\coordinatorIndex \in \{1,\dots,\numCoordinators\}$, find a max-consensus in $\networkGraph_\coordinatorIndex$ and update the inputs of all agents (to be used in all future max-consensus steps) to be the consensus value.
 \item Repeat step~\ref{extended-network-scheme-step} a total of $\numCoordinators$ times.
\end{enumerate}
After the initial execution of step~\ref{extended-network-scheme-step}, at least one subgraph of agents will have the true maximum as the input for future consensus schemes. The connectivity requirement ensures that after each further execution of step~\ref{extended-network-scheme-step}, this property is propagated to at least one additional subgraph, so after $\numCoordinators$ repetitions, the whole network has achieved max-consensus.

Note that the ScalableMax or ScalableMax-EC scheme is executed a total of $\numCoordinators^2$ times, so our scheme can be advantageous compared to the random-pairwise or random-broadcast scheme only as long as the network can be partitioned into subgraphs of very large size with a very small number of coordinators, which can for example be the case in ultra-dense networks of not overly large diameter.

\section{Conclusion}
\label{section:conclusion}
We have introduced a novel max-consensus protocol designed to handle noise while exploiting interference in order to be highly scalable in star-shaped wireless networks. Under minimal assumptions on the initial values, we have proved analytically that the consensus is reached with complexity that is logarithmic in the number of agents.
For the low and medium SNR regime, we have added an error correction mechanism which achieves lower overall error at the expense of increased complexity. Our simulations have demonstrated that logarithmic complexity is retained and the proposed schemes compare favorably with state-of-the-art baselines if the network is dense. Finally, we have extended the proposed schemes to more general, non-star-shaped networks. Open questions for future research include finding a mechanism for distributed clustering, how to initialize the scheme with an optimal starting sequence and finding suitable pre- and post-processing schemes that deal with fading.
\bibliographystyle{IEEEtran}
\bibliography{IEEEabrv,references}

\end{document}